\documentclass[10pt,journal,final]{IEEEtran}
\usepackage[utf8]{inputenc}
\usepackage{bm}
\usepackage{amsthm}
\usepackage{amsmath}
\usepackage{amssymb}
\usepackage{color}
\usepackage{graphicx}
\usepackage{epstopdf}
\usepackage{cite}
\usepackage[ruled]{algorithm2e}
\newtheorem{definition}{Definition}
\newtheorem{lemma}{Lemma}

\newtheorem{theorem}{Theorem}

\begin{document}
	\title{
		Direction Finding of Electromagnetic Sources on a Sparse Cross-Dipole Array Using One-Bit Measurements
	}
	\author{Zhiyong Cheng,
		Shengyao Chen,~\IEEEmembership{Member,~IEEE,}
		Qibin Shen,\\
		Jin He,\\
		and Zhong Liu,~\IEEEmembership{Member,~IEEE,}
		\thanks{This work was supported in part by the National Natural Science Foundation of China under Grants 61671245, 61401210, 61571228.}
		\thanks{Z. Cheng, C. Chen, and Z. Liu are with the School of Electronic and Optical Engineering, Nanjing University of Science and Technology, Nanjing, China. (e-mail:chengzhiyong@njust.edu.cn; chenshengyao@njust.edu.cn)}%
		\thanks{Qibin Shen is with the Department of Mathematics, University of Rochester, Rochester, NY 14642, USA.}
		\thanks{J. He is with the Shanghai Key Laboratory of Intelligent Sensing and Recognition, Department of Electronic Engineering, Shanghai Jiaotong University, Shanghai 200240, China.}}
	\maketitle

	
\begin{abstract}
Sparse array arrangement has been widely used in vector-sensor arrays because of increased degree-of-freedoms for identifying more sources than sensors.
For large-size sparse vector-sensor arrays, one-bit measurements can further reduce the receiver system complexity by using low-resolution ADCs.
In this paper, we present a sparse cross-dipole array with one-bit measurements to estimate Direction of Arrivals (DOA) of electromagnetic sources.
Based on the independence assumption of sources, we establish the relation between the covariance matrix of one-bit measurements and that of unquantized measurements by Bussgang Theorem.
Then we develop a Spatial-Smooth MUSIC (SS-MUSIC) based method, One-Bit MUSIC (OB-MUSIC), to estimate the DOAs.
By jointly utilizing the covariance matrices of two dipole arrays, we find that OB-MUSIC is robust against polarization states.
We also derive the Cramer-Rao bound (CRB) of DOA estimation for the proposed scheme.
Furthermore, we theoretically analyze the applicability of the independence assumption of sources, which is the fundamental of the proposed and other typical methods, and verify the assumption in typical communication applications.
Numerical results show that, with the same number of sensors, one-bit sparse cross-dipole arrays have comparable performance with unquantized uniform linear arrays and thus provide a compromise between the DOA estimation performance and the system complexity.
\end{abstract}
\begin{IEEEkeywords}
	Cross-dipole, sparse arrays, one-bit measurements, Cramer-Rao bound, DOA estimation
\end{IEEEkeywords}

\section{INTRODUCTION}
Signal processing for electromagnetic (EM) signals has attracted attention as the reason of carrying more information than traditional signals in the past decades \cite{nehorai1994vector,li1991angle,han2014nested}.
Many array signal processing techniques have been developed for direction of arrival (DOA) estimation using vector-sensor arrays designed for EM signals.
Benefitting from increased degrees of freedom (DOF), sparse array arrangement is widely used in both scalar and vector sensor arrays.
Specifically, there are three popular kinds of sparse arrays, minimum-redundancy arrays (MRA) \cite{moffet1968minimum}, nested arrays \cite{pal2010nested} and coprime arrays \cite{pal2011coprime}.
All of them can identify $O(N^2)$ sources with only $O(N)$ sensors \cite{pal2012application} \cite{zhang2013sparsity} by exploiting the second order statistics.
Unlike MRA, the latter two arrays are more widely used in DOA estimation during the past decade for their straightforward closed mathematical expression.

Recently, a considerable amount of articles has been published about sparse vector-sensor arrays.
For completely polarized (CP) signals, a tensor-based model using nested array has been proposed in \cite{han2014nested}, multiple parameters estimation has been studied in \cite{han2018polarization} and \cite{dong2016two}, mutual coupling reduction has been studied in \cite{yang2019spatially}, sparse reconstruction has been discussed with coprime array in \cite{si2019three} and \cite{li2018polarization}.
For partially polarized (PP) signals, \cite{tao2017stokes} provided a method to estimate Stokes vectors, \cite{he2017direction} and \cite{chang2017stokes} discussed the DOA estimation for nested arrays and coprime arrays, respectively.
From the viewpoint of hardware cost, sparse arrays are more economical than uniform linear arrays (ULA), especially when the source number is larger than the sensor number. Moreover, sparse arrays have less mutual coupling than ULAs.

One-bit measurement is another way to reduce hardware cost by recording received signals with one-bit analog-to-digital converters (ADC) \cite{boufounos20081} \cite{jacovitti1994estimation}.
Although source signals are difficult to recover, some meaningful parameters can be estimated by new methodologies in many applications, such as pulse-Doppler radars \cite{xi2018super} \cite{ren2019sinusoidal} and massive MIMO \cite{risi2014massive} \cite{mollen2016uplink}.
One-bit measurement has also been attracting a lot of interest in scalar array processing.
For ULAs, \cite{stockle20151} exploited sparse reconstruction techniques to estimate the DOA, and \cite{huang2019one} analyzed the MUSIC algorithm with one-bit covariance matrix.
For sparse arrays, \cite{liu2017one} presented a SS-MUSIC-based method for nested and co-prime arrays, and \cite{Chen2018} utilized sparse reconstruction methods for compressive sparse array.
The performance bounds of DOA estimation using one-bit measurements are derived in \cite{bar2002doa} \cite{stein2016doa}.
However, all these studies are based on scale sensor arrays. So far, little attention has been paid to one-bit measurement for vector sensor arrays to our knowledge.

In this paper, we focus on the DOA estimation using one-bit measurements for sparse vector-sensor arrays, since sparse array arrangement increases DOFs and one-bit measurements reduce the cost of ADCs.
Even though the original covariance matrix of received signals cannot be recovered from one-bit measurements, a constant relationship called \textbf {arcsin law} \cite{van1966spectrum} makes a bridge between the normalized original covariance matrix and the covariance matrix of one-bit measurements when source signals obey independent Gaussian distribution.
By recovering the normalized original covariance matrix, we provide two SS-MUSIC-based approaches to estimate DOAs.
The first one is a natural extension of that for scalar sensor arrays, while another one is specifically designed for vector-sensor arrays such that it is robust for both PP and CP signals.
We derive the Cramer-Rao bound (CRB) of DOA estimation to understand the effect of one-bit measurements for the proposed scheme.
Moreover, we propose a theoretical analysis on the independence of EM signals, which is the basic premise of our and other typical DOA estimation methods \cite{he2017direction,chang2017stokes,tao2017stokes}.
Even the premise has been widely used in lots of articles, its theoretical analysis is still missing.
We discuss the rationality of the premise in typical communication applications.

Main contributions of this paper are summarized as follows:
\begin{description}
	\label{contributions}
	\item[1)] A cross-dipoles sparse array scheme using one-bit measurements is proposed. The DOF are increased by sparse arrangement of cross-dipoles and the sampling complexity is reduced by one-bit measurements.
	\item[2)] A subspace algorithm called OB-MUSIC is developed for the proposed array to estimate DOAs. By exploiting the structure of the covariance matrix, this algorithm is robust to both PP and CP signals and its performance is comparable with that of unquantized measurements on ULAs.
	\item[3)] The CRB of one-bit sparse cross-dipole arrays is derived. We find that the Fisher Information Matrix (FIM) for one-bit measurements is a weighted version of that for the unquantized measurements.
	\item[4)] A theoretical analysis is performed to declare that it is suitable to assume that all the source signals are independent to each other in typical communication applications.
\end{description}

The rest of this paper is organized as follows. Original cross-dipoles array, the properties of EM signals and sparse arrays are reviewed in Section \ref{SM}. Then in Section \ref{OCSADOAE}, one-bit cross-dipoles sparse array scenario is given, followed by the proposed DOA estimation algorithms.The CRB is derived in Section \ref{CRB}. The theoretical analysis on the independence of EM signals is shown in Section \ref{OREMS}. Section \ref{NR} presents numerical results and Section \ref{Conclusion} concludes the paper.

\emph{Notations:}
In this paper, scalars are denoted by lowercase italic letters, vectors are denoted by lowercase bold letters, and matrices are denoted by uppercase bold letters.
Sets are denoted by uppercase hollow letters, e.g. $\mathbb{C}$.
The superscripts $*$, $T$ and $H$ denote the complex conjugate, transpose and Hermitian transposition, respectively.
The superscripts $\dagger$ denote the Moore-Penrose pseudoinverse.
$\mathbb{E}\lbrace\cdot\rbrace$, $\textrm{diag}(\cdot)$, $\textrm{trace}(\cdot)$, $\textrm{det}(\cdot)$ and $\textrm{vec}(\cdot)$ denote the expectation operator, creating diagonal matrix operator, getting trace operator, determinant operator and vectorization operator.
Angle brackets with subscripts denote getting elements of a matrix, e.g. $\langle\mathbf{A}\rangle_{m,n}$ means the element of $\mathbf{A}$ on location row $m$ and column $n$ as a scalar, and $\langle\mathbf{A}\rangle_{m,:}$ means the whole row $m$ of A as a vector.
$\mathbf{I}_L$ denotes the $L\times L$ identity matrix.
$\mathbf{0}_L$ denotes the $L\times L$ zero matrix.

\section{Signal Model}
\label{SM}
\subsection{Receive Signal Model}

Consider $K$ EM waves travelling in an isotropic and homogeneous medium and impinging on a one-dimensional array consisting of $L$ cross-dipoles, as shown in Fig. \ref{crossdipole}, where the $l$-th cross-dipole is located at the position $\omega_l\lambda/2$. Here $\lambda$ is the wavelength and $\omega_l$ is an element of set $\mathbb{S}=\lbrace \omega_1,\omega_2,..\omega_L\rbrace$ denoting the array arrangement (e.g., $\omega_l = l$ for ULA). As in \cite{li1994efficient}, each cross-dipole consists of an $x$-axis dipole and a $y$-axis dipole paralleling to $x$-axis and $y$-axis, respectively\cite{he2017direction}.
\begin{figure}[h]
	\centering
	\includegraphics[scale=0.5]{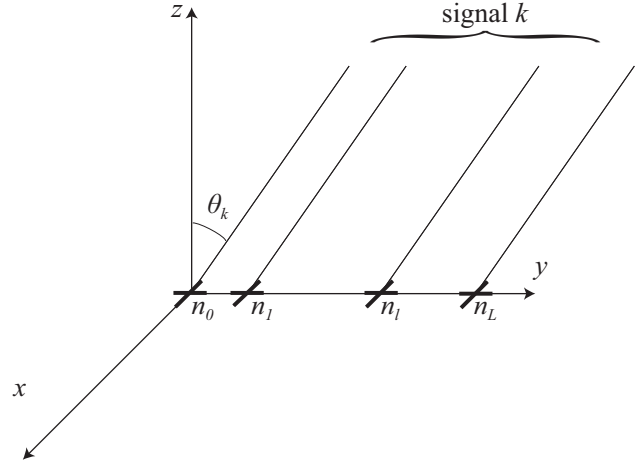}
	\caption{Cross-dipoles array.}
	\label{crossdipole}
\end{figure}

The $l$-th cross-dipole measurements ${\mathbf x}_{l} = \lbrack x_{l,1}(t),x_{l,2}(t)\rbrack^T$ can be modeled as
\begin{equation}
\label{measurementsModel}
{\mathbf x}_{l}(t)=\sum_{k=1}^Kv_l({\overline\theta}_k){\mathbf B}_{k}{\mathbf s}_k(t)+{\mathbf n}_{l}(t),
\end{equation}
where ${\mathbf B}_k=\textrm{diag}(-1,\cos(\arcsin{2\overline\theta}_k))$, ${\overline\theta}_k=\sin\theta_k/2\in\lbrack-1/2,1/2\rbrack$ and $\theta_k\in\lbrack-\pi/2,\pi/2\rbrack$ denote the cross-dipole response, the normalized DOA and the DOA of the $k$-th source, respectively, $v_l({\overline\theta}_k)=e^{j2\pi{\overline\theta}_k \omega_l}$ is the spatial response of the $l$-th cross-dipole with $\omega_l\in\mathbb{S}$,
${\mathbf s}_k(t)=\lbrack s_{k,1}(t),s_{k,2}(t)\rbrack^T$ is the signal vector, and ${\mathbf n}_{l}(t)=\lbrack n_{l,1}(t),n_{l,2}(t)\rbrack^T$ is the additive noise vector.

The covariance matrix of $\mathbf{s}_k(t)$ is given by  \cite{li1994efficient}
\begin{equation}
\label{signalcovariance}
\begin{aligned}
\mathbf{R}_{s_k} &= \mathbb{E}\left\lbrace \mathbf{s}_k(t)\mathbf{s}_k^H(t)\right\rbrace \\
&=\mathbb{E}\lbrace
\begin{bmatrix}
s_{k,1}(t)s_{k,1}^*(t)&s_{k,1}(t)s_{k,2}^*(t)\\
s_{k,2}(t)s_{k,1}^*(t)&s_{k,2}(t)s_{k,2}^*(t)
\end{bmatrix}\rbrace
\\
&\triangleq
\begin{bmatrix}
r_{k,11}&r_{k,12}\\
r_{k,12}^*&r_{k,22}
\end{bmatrix}
=p^2_k
\begin{bmatrix}
\rho_{k,11}&\rho_{k,12}\\
\rho_{k,12}^*&\rho_{k,22}
\end{bmatrix}\\
&=p^2_k\overline{\mathbf{R}}_{s_k},
\end{aligned}
\end{equation}
where $p^2_k$ and $\overline{\mathbf{R}}_{s_k}$  denote the signal power and the normalized covariance matrix, respectively.
As proposed in \cite{nehorai1994vector}, each EM signal has two spatial DOFs presented in a EM wave, so it can carry two independent signals and transmit them simultaneously.
In the Dual Signal Transmission (DST) method as in \cite{nehorai1994vector}, two independent signals are transmitted and thus the covariance matrix $\mathbf{R}_{s_k}$ is of full rank. On the other hand, in the Single Signal Transmission (SST) method, only one signal is transmitted so that the covariance matrix $\mathbf{R}_{s_k}$ is singular.
In other words, DST makes full use of two spatial DOFs while SST uses only one.

The degree of polarization (DOP) of $\mathbf{s}_k(t)$ is defined as \cite{li1994efficient}
\begin{equation}
\label{DOF}
\begin{aligned}
\eta_k &= \left[ 1-\frac {4\det(\mathbf{R}_{s_k})}{\left[ \textrm{ trace}(\mathbf{R}_{s_k})\right] ^2}\right]^{1/2} \\
&=\left[ 1-4(\rho_{k,11}\rho_{k,22}-{\left|\rho_{k,12} \right|}^2 )\right]^{1/2},
\end{aligned}
\end{equation}
with $\eta_k\in\left[ 0,1\right] $. The EM signal is completely polarized (CP) with $\eta_k = 1$ and partially polarized (PP) with $\eta_k<1$. Especially, the EM signal is unpolarized (UP) when $\eta_k = 0$.

Generally speaking, an EM signal can be decomposed as a sum of a CP part and an UP part \cite{he2017direction}, and then the DOP $\eta_k$ can be expressed as $\eta_k = {p^2_{k,c}}/({p^2_{k,c}}+{p^2_{k,u}})$, where ${p^2_{k,c}}$ and ${p^2_{k,u}}$ denote the power of the CP part and the UP part, respectively.
With this decomposition, the covariance matrix $\mathbf{R}_{s_k}$ can be expressed as
\begin{equation}
\label{CPdecomposition}
\mathbf{R}_{s_k}=p^2_{k,c}{\mathbf{Q}({\alpha_k})}{\mathbf{w}({\beta_k})}{\mathbf{w}^H({\beta_k})}{\mathbf{Q}^H({\alpha_k})}+\frac{p^2_{k,u}}{2}{\mathbf{I}_2},
\end{equation}
where
\begin{equation}
\label{Qwmatrix}
\begin{aligned}
\mathbf{Q}({\alpha_k}) &=
\begin{bmatrix}
\cos{\alpha_k}&\sin{\alpha_k}\\
-\sin{\alpha_k}&\cos{\alpha_k}
\end{bmatrix},
\\
\mathbf{w}({\beta_k}) &=
\begin{bmatrix}
\cos{\beta_k}\\
j\sin{\beta_k}
\end{bmatrix},
\end{aligned}
\end{equation}
with $\alpha_k\in\left[ -\pi/2,\pi/2\right] $ and $\beta_k\in\left[ -\pi/4,\pi/4\right] $ denoting the polarization orientation angle and polarization ellipticity angle, respectively.
Especially, when the EM signal is CP, the Jones vector is used to describe the polarization state and $\mathbf{s}_k(t)$ is rewritten as \cite{he2013near}
\begin{equation}
\label{JonesVector}
\mathbf{s}_k(t) =
\begin{bmatrix}
\cos{\varphi_k}\\
\sin{\varphi_k}e^{j\psi_k}
\end{bmatrix}s_k(t) = {\bm {\mathcal{J}}_k}s_k(t),
\end{equation}
where $\varphi_k\in\left[0,\pi/2 \right] $ and $\psi_k\in\left( -\pi,\pi\right] $ are the polarization parameters denoting the auxiliary polarization angle and the auxiliary polarization phase difference, respectively, and ${\bm{\mathcal{J}}_k} = {\left[ \cos{\varphi_k},\sin{\varphi_k}e^{j\psi_k}\right]}^T $ denotes the normalized Jones vector \cite{li2018polarization}.

For all $K$ received signals, denote the signals vector as
\begin{equation}
\label{sigvector}
\mathbf{s}_{\mathbb{K}}(t) = [\mathbf{s}^T_{1}(t),...,\mathbf{s}^T_K(t)]^T \in \mathbb{C}^{2K},
\end{equation}
where ${\mathbf s}_k(t)=\lbrack s_{k,1}(t),s_{k,2}(t)\rbrack^T$ and the signals vector received by dipoles on each axis as
\begin{equation}
\label{sigvector2}
\mathbf{s}_{\mathbb{K},m}(t) = [s_{1,m}(t),...,s_{K,m}(t)]^T \in \mathbb{C}^{K},
\end{equation}
where $m=1$ and $m=2$ means the $x$-axis and the $y$-axis, respectively.

In this paper, we use a stochastic model to describe the received signals. We propose the sparse array models and DOA estimation algorithms under the following assumptions:
\begin{description}
	\label{assumption}
	\item[1)] All the $K$ sources are independent random Gaussian processes.
	\item[2)] The DOA of each source is different from the other.
	\item[3)] The noises follow independent complex Gaussian distribution $ \mathcal{CN}(0,\sigma^2\mathbf{I}_{2K})$.
	\item[4)] The noises are statistically independent to the sources.
\end{description}

The assumption 1) has widely been used for decades.
It is also a key assumption for DOA estimation methods developed in the next section, but it has rarely been analyzed.
In Section \ref{OREMS}, we will propose a theoretical analysis declaring that this assumption is suitable for typical communication application.

\subsection{Difference Coarray}
The $x$-axis and $y$-axis dipoles measurements can be separately modeled as \cite{he2013near}
\begin{equation}
\label{dipolesmodel}
\mathbf{x}_{{\mathbb{S}},m}(t)= \mathbf{A}_{{\mathbb{S}}}\overline{\mathbf{B}}_m\mathbf{s}_{\mathbb{K},m}(t)+\mathbf{n}_{{\mathbb{S}},m}(t),
\end{equation}
where
\begin{equation}
\label{dipolesmodel2}
\begin{aligned}
m \in\lbrace1,2\rbrace&,\; \mathbb{K}=\lbrace1,...,K\rbrace,\\
\mathbf{x}_{{\mathbb{S}},m}(t)&= \left[ x_{1,m}(t),...,x_{L,m}(t)\right]^T,\\
\mathbf{A}_{{\mathbb{S}}}&=\left[\mathbf{v}_{{\mathbb{S}}}({\overline\theta}_1),...,\mathbf{v}_{{\mathbb{S}}}({\overline\theta}_K) \right], \\
\mathbf{v}_{{\mathbb{S}}}({\overline\theta}_k)&= \left[ v_1({\overline\theta}_k),...,v_L({\overline\theta}_k)\right] ^T,\\
\overline{\mathbf{B}}_m &= \textrm {diag}([b_m(\overline{\theta}_1),...,b_m(\overline{\theta}_K)]), \\
b_m(\overline{\theta}_k) &=\left\lbrace
\begin{array}{rcl}
-1,&\textrm {if}\;m=1;\\
\cos({2\arcsin{\overline{\theta}_k}}),& \textrm{if}\;m=2,
\end{array} \right. \\
\mathbf{n}_{{\mathbb{S}},m}(t)&=\left[ n_{1,m}(t),...,n_{L,m}(t)\right]^T .
\end{aligned}
\end{equation}
In the above equations, $\mathbf{A}_{{\mathbb{S}}}$ denotes the $L\times K$ array steering matrix, $\mathbf{s}_{\mathbb{K},m}(t)$ and $\mathbf{n}_{{\mathbb{S}},m}(t)$ are the signals vector and the noise vector received by dipoles on the $x$-axis or $y$-axis, respectively.

Consider the signals received by dipoles on the $m$-th axis, the covariance matrix of $\mathbf{x}_{{\mathbb{S}},m}(t)$ is given as
\begin{equation}
\label{xdipolesC}
\begin{aligned}
\mathbf{R}_{\mathbf{x}_{{\mathbb{S}},m}} =& \mathbf{A}_{{\mathbb{S}}}\overline{\mathbf{B}}_m\mathbb{E}\lbrace \mathbf{s}_{{\mathbb{K}},m}(t)\mathbf{s}^*_{{\mathbb{K}},m}(t) \rbrace \overline{\mathbf{B}}^H_m\mathbf{A}_{{\mathbb{S}}}^H+\sigma^2\mathbf{I}_L\\
=& \mathbf{A}_{{\mathbb{S}}}\mathbf{P}_m\mathbf{A}_{{\mathbb{S}}}^H+\sigma^2\mathbf{I}_L
\end{aligned}
\end{equation}
where $\mathbf{P}_m =  \textrm{diag}(b^2_m(\overline{\theta}_1)r_{1,mm},...,b^2_m(\overline{\theta}_K)r_{K,mm})$ and $\sigma^2$ is the noise power.
Vectorizing and combining duplicate entries in (\ref{xdipolesC}) leads to the following vector
\begin{equation}
\label{vectorize}
\mathbf{x}_{\mathbb{D},m} =\mathbf{W}^\dagger\textrm{vec} (\mathbf{R}_{\mathbf{x}_{{\mathbb{S}},m}})=\mathbf{A}_{\mathbb{D}}\mathbf{p}_m+\sigma^2\mathbf{e}_0.
\end{equation}
In (\ref{vectorize}), $\mathbb{D}$ is the difference coarray defined as in definition 1.
\begin{definition}[Difference coarray]
	Assume an integer set $\mathbb{S}$ denoting the sensor locations, its difference coarray is defined as $\mathbb{D} = \left\lbrace \omega_i-\omega_j\mid \omega_i,\omega_j\in\mathbb{S} \right\rbrace $.
\end{definition}
\noindent The coarray steering matrix $\mathbf{A}_{\mathbb{D}}$, the visual signal $\mathbf{p}_m$ and the normalize visual noise vector $\mathbf{e}_0$ are expressed as
\begin{equation}
\label{coarraysteering}
\begin{aligned}
\mathbf{A}_{\mathbb{D}} &= \left[\mathbf{v}_{\mathbb{D}}({\overline\theta}_1),...,\mathbf{v}_{\mathbb{D}}({\overline\theta}_K) \right], \\
\mathbf{p}_m &= [ b^2_m(\overline{\theta}_1)r_{1,mm},...,b^2_m(\overline{\theta}_K)r_{K,mm}] ^T,\\
\left\langle \mathbf{e}_0\right\rangle_{d} &= \delta_{d,0},
\end{aligned}
\end{equation}
where $\left\langle \mathbf{v}_{\mathbb{D}}({\overline\theta}_k)\right\rangle_{d}  =v_{i}({\overline\theta}_k) v_{j}^*({\overline\theta}_k)$, $d\in\mathbb{D}$ , $d = \omega_i-\omega_j$ and $\delta_{d,0}$ is the Kronecker delta.
The matrix $\mathbf{W}$ is defined as
\begin{definition}[The matrix $\mathbf{W}$ \cite{liu2017one}]
	The binary matrix  $\mathbf{W}$ has size $\left|\mathbb{S} \right|^2 $ -by-  $\left|\mathbb{D} \right| $. The columns of $\mathbf{W}$ satisfy $\left\langle \mathbf{W}\right\rangle_{:,d} = [\textrm{vec}(\mathbf{J}(d))]^T$ for $d\in\mathbb{D}$, where $\left\langle \mathbf{J}(d)\right\rangle _{\omega_i,\omega_j}\in\left\lbrace0,1 \right\rbrace ^{\left| \mathbb{S}\right|\times\left| \mathbb{S}\right| }$ is given by
	\begin{equation}\nonumber
	\label{J}
	\left\langle \mathbf{J}(d)\right\rangle _{\omega_i,\omega_j}=\left\{
	\begin{array}{rcl}
	1,&  &\text{if}\;\omega_i-\omega_j = d,\\
	0,&  &\text{otherwise}.
	\end{array}
	\right.\forall\omega_i,\omega_j\in\mathbb{S}.
	\end{equation}
\end{definition}

After vectorization and combination, $\mathbf{x}_{\mathbb{D},m}$ is considered as a measurement generated by the coarray with steering matrix $\mathbf{A}_{\mathbb{D}}$.
Although the visual signals denoted by the power of real signals are coherent, the DOAs can be estimated by subspace methods such as SS-MUSIC.
With proper array arrangement, the DOFs of difference coarray achieve $O(N^2)$ by using only $O(N)$ sensors, so that the number of resolved sources is much larger than that of the sensors.

Nested arrays and coprime arrays are two popular sparse array arrangements for their straightforward closed mathematical expression.
In order to employ the subspace algorithm, the longest uniform part of $\mathbb{D}$ expressed as $\mathbb{U}$ is selected.
For a nested array with $L_1+L_2$ cross-dipoles, as shown in Fig. \ref{arrays}(a), the congfiguration is given by \cite{pal2010nested}
\begin{equation}
\label{nestedarrays}
{\mathbb{S}}_{n}  = \lbrace 1,2,3,...L_1,L_1+1,...,L_2(L_1+1)\rbrace.
\end{equation}
The difference coarray of nested array is a ULA, where $\mathbb{D}_{n}=\mathbb{U}_{n}=\left\lbrace0,\pm1,\pm2,...,\pm{L_2(L_1+1)-1} \right\rbrace $.
But for coprime arrays, as in Fig. \ref{arrays}(b), $2M+N-1$ cross-dipoles and the sensor are located as
\begin{equation}
\label{coprimearrays}
{\mathbb{S}}_{c}  = \lbrace 0,M,...,(N-1)M,N,...,(2M-1)N \rbrace.
\end{equation}
Then we have $\mathbb{U}_{c}=\left\lbrace0,\pm1,\pm2,...,\pm MN+M-1 \right\rbrace $ which is smaller than $\mathbb{D}_{c}$, because there are some holes on the difference coarray.
For instance, in Fig. \ref{arrays}(b), $\mathbb{D}_{c}$ does not contain elements $\pm8$.

\begin{figure}[h]
	\centering
	\includegraphics[scale=0.5]{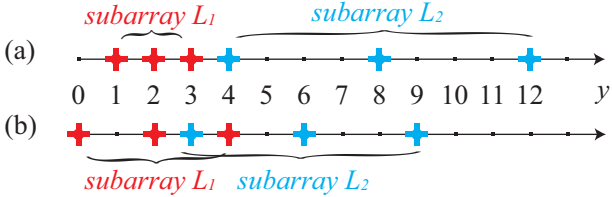}
	\caption{The configuration of (a) a nested array with $L_1 = L_2 = 3$ and (b) a coprime array with $M=2$ and $N=3$.}
	\label{arrays}
\end{figure}

\section{One-bit Cross-dipoles Sparse Array DOA Estimation}
\label{OCSADOAE}
\subsection{One-bit Measurements of Cross-dipoles Sparse Array}
One-bit measurements are particularly useful in reducing system cost of ADCs.
In the proposed array, two one-bit quantizers are employed on each dipole to quantize complex signals.
So each one-bit quantization cross-dipole has four one-bit quantizers to quantize two orthogonal parts of complex EM signals.
Then, the one-bit cross-dipoles sparse array is constructed by one-bit cross-dipoles.
For instance, a one-bit cross-dipoles nested array is shown in Fig. \ref{onebitnestedarray}.

 \begin{figure}[h]
	\centering
	\includegraphics[scale=0.5]{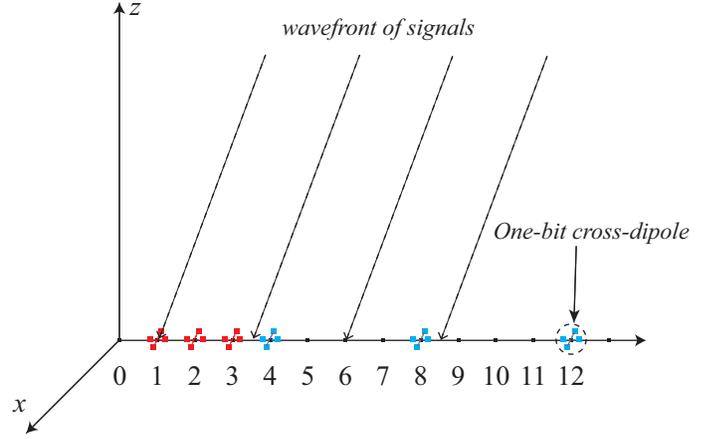}
	\caption{The configuration of one-bit cross-dipoles nested array with $6$ one-bit cross-dipoles (i.e. $24$ one-bit sensors ) and $L_1 = L_2 = 3$}
		\label{onebitnestedarray}
	\end{figure}

A complex number $c\in \mathbb{C}$ quantized by a one-bit dipole can be represented as a $\textrm {signe}(\cdot)$ operator acting on it, where the $\textrm {signe}(\cdot)$ operator is defined as \cite{Chen2018}

\begin{equation}
\label{signc}
\textrm{signe}(c) = \frac{1}{\sqrt{2}}( \textrm {sign}( \mathcal{R}(c)) +j\textrm {sign}(\mathcal{I}(c)) ),
\end{equation}
where $\textrm {sign}(\cdot)$ is the sign function acting on a real number $a\in\mathbb{R}$ expressed as follows
\begin{equation}
\label{sign}
\textrm {sign}(a) =  \left\{
\begin{array}{rcl}
1,&  &\text{if}\;a>0,\\
-1,&  &\text{otherwise},
\end{array}\right.
\end{equation}
and $\mathcal{R}(c)$ and $\mathcal{I}(c)$ get the real part and the imaginary part of $c$, respectively.

With (\ref{signc}), one-bit cross-dipoles array measurements can be expressed as
\begin{equation}
\label{oneebitmeasurement}
\mathbf{y}_{{\mathbb{S}},m}(t) = \textrm {signc}(\mathbf{x}_{{\mathbb{S}},m}(t)).
\end{equation}
Then the covariance matrix of $\mathbf{y}_{{\mathbb{S}},m}(t)$ is
\begin{equation}
\label{cofoneebit}
\mathbf{R}_{\mathbf{y}_{\mathbb{S},m}} = \mathbb{E}\lbrace \mathbf{y}_{\mathbb{S},m}(t)\mathbf{y}^H_{\mathbb{S},m}(t)\rbrace.
\end{equation}
Now, we focus on the issue how to estimate DOAs via the one-bit covariance matrix $\mathbf{R}_{\mathbf{y}_{\mathbb{S},m}}$.
Researches over the past decade provided a lot of powerful methods on DOA estimation with the unquantized covariance matrix $\mathbf{R}_{\mathbf{x}_{\mathbb{S},m}}$.
In the following subsections, we will provide a one-to-one mapping between the unquantized and one-bit covariance matrices and then develop corresponding DOA estimation algorithms.

\subsection{Reconstruction of Unquantized Covariance Matrix}

We rewrite the separately received EM signals $\mathbf{s}_{\mathbb{K},m}(t)$ in (\ref{sigvector2}) as
\begin{equation}
\label{xyaxis}
\mathbf{s}_{\mathbb{K},m}(t) =\left[ p_{1,m}\overline{s}_{1,m}(t),...,p_{K,m}\overline{s}_{K,m}(t)\right]^T,
\end{equation}
where $p^2_{k,m}$ and $\overline{s}_{k,m}(t)$ denote the power and the normalized signal of $s_{k,m}(t)$, respectively, and $p^2_{k,m}=r_{k,mm}$.
The covariance matrix of $\mathbf{x}_{{\mathbb{S}},m}(t)$ in (\ref{xdipolesC}) is rewritten as
\begin{equation}
\label{cmx}
\mathbf{R}_{\mathbf{x}_{{\mathbb{S}},m}} =\sum_{k=1}^Kp^2_{k,m}b^2_m(\overline{\theta}_k)\mathbf{v}_{{\mathbb{S}}}(\overline{\theta}_k)\mathbf{v}_{{\mathbb{S}}}^H(\overline{\theta}_k)+\sigma^2\mathbf{I}_L.
\end{equation}
The normalized covariance matrix of $\mathbf{x}_{{\mathbb{S}},m}(t)$ is then defined as
\begin{equation}
\label{ncmx}
\overline{\mathbf{R}}_{\mathbf{x}_{{\mathbb{S}},m}} = \mathbf{N}_m^{-1/2}\mathbf{R}_{\mathbf{x}_{{\mathbb{S}},m}}\mathbf{N}_m^{-1/2},
\end{equation}
where $\mathbf{N}_m$ is a diagonal matrix with $\langle \mathbf{N}_m\rangle_{l,l} = \langle \mathbf{R}_{\mathbf{x}_{{\mathbb{S}},m}}\rangle_{l,l}$.

When the sources obey independent and identically distributed (i.i.d.) Gaussian distribution, the relation between $\mathbf{R}_{\mathbf{y}_{{\mathbb{S}},m}}$ and $\overline{\mathbf{R}}_{\mathbf{x}_{{\mathbb{S}},m}}$ was established by the arcsin law \cite{van1966spectrum} and Bussgang Theorem \cite{bussgang1952crosscorrelation} as
\begin{equation}
\label{obc}
\langle\mathbf{R}_{\mathbf{y}_{{\mathbb{S}},m}}\rangle_{i,j} = \frac{2}{\pi}\textrm {arcsine}(\langle\overline{\mathbf{R}}_{\mathbf{x}_{{\mathbb{S}},m}}\rangle_{i,j}),
\end{equation}
where the $\textrm {arcsine}(\cdot)$ operator is defined as
\begin{equation}\nonumber
\textrm{arcsine}(c) = \arcsin(\mathcal{R}(c))+j\arcsin(\mathcal{I}(c)).
\end{equation}
Therefore, the normalized covariance matrix can be reconstructed by using (\ref{obc}), expressed as
\begin{equation}
\label{obc2}
\langle\overline{\mathbf{R}}_{\mathbf{x}_{{\mathbb{S}},m}}\rangle_{i,j} = \textrm {sine}(\frac{\pi}{2}\langle\mathbf{R}_{\mathbf{y}_{{\mathbb{S}},m}}\rangle_{i,j}),
\end{equation}
where
\begin{equation}\nonumber
\textrm {sine}(c) = \sin(\mathcal{R}(c))+j\sin(\mathcal{I}(c)).
\end{equation}

It is worth pointing out that the data on each axis are jointly processed for unquantized measurements \cite{yang2019spatially,he2017direction}.
However, for one-bit measurements, we reconstruct the normalized covariance matrix of each axis separately since the data in each axis is suitable for arcsin law and Bussgang theorem.
Due to nonlinear one-bit sampling, the joint recovery of covariance cannot be achieved directly.
In the future studies, we will seek to reconstruct the covariance matrix jointly or estimate the parameters without reconstruction.

The relation between the unquantized covariance matrix and its normalized form is declared by the following lemma.
\begin{lemma}[Lemma $1$ in \cite{liu2017one}]
	\label{norlemma}
	If the sources are all independent to each other, we have $\mathbf{R}_{\mathbf{x}_{{\mathbb{S}},m}} = N_m\overline{\mathbf{R}}_{\mathbf{x}_{{\mathbb{S}},m}}$, where $N_m = \sum^K_{k=1}p^2_{k,m}b^2_m(\overline{\theta}_k)+\sigma^2 >0$.
\end{lemma}

Lemma \ref{norlemma} demonstrates that the normalized covariance matrix is obtained by scaling the unquantized covariance matrix with a positive number $N_m$.
The two covariance matrices share the same vector space, and thus the DOAs can be estimated from $\overline{\mathbf{R}}_{\mathbf{x}_{{\mathbb{S}},m}}$ by using subspace methods.

In fact, $N_m$ is the total power been received by the $m$-th axis, which is lost during the one-bit measurements.
Then we cannot estimate the parameters that are related to the ratio of signals power on the two axis such as the polarization orientation angle or the polarization ellipticity angle.
However, the unknown $N_m$ has no effect on the DOA estimation since the following proposed method is based on the normalized covariance matrix $\overline{\mathbf{R}}_{\mathbf{x}_{{\mathbb{S}},m}}$.
If $N_m$ is required, we can utilize time-varying threshold-based one-bit measurement methods as in \cite{xi2020gridless} and \cite{fu2018quantized}.

\subsection{DOA Estimation Using SS-MUSIC}
In this subsection, we propose two DOA algorithms based on SS-MUSIC, the first one called OB-MUSIC1 estimates DOA with $\mathbf{R}_{\mathbf{y}_{{\mathbb{S}},1}}$ or $\mathbf{R}_{\mathbf{y}_{{\mathbb{S}},2}}$ separately, while the second one called OB-MUSIC2 solves the DOA estimation with a combination of the two covariance matrices.
The OB-MUSIC mentioned above refers specifically to OB-MUSIC2.

\begin{algorithm}
	\label{algorithm1}
	\caption{OB-MUSIC1}
	\KwIn{$\mathbf{R}_{\mathbf{y}_{{\mathbb{S}},m}}$, number of signals $K$, array configration $\mathbb{D}$}
	\KwOut{normalized DOAs of signals $ \overline{\theta}_k $}
	1) Consider $\mathbf{R}_{\mathbf{y}_{{\mathbb{S}},m}}$ for $m=1$ or $m=2$ separately and select one of the axis.
	 As in (\ref{vectorize}), the coarray measurement been selected is
	\begin{equation}
	\label{music1vv}
	\overline{\mathbf{x}}_{\mathbb{D},m} = \mathbf{W}^\dag\textrm{vec}(\textrm{sine}(\frac{\pi}{2}\mathbf{R}_{\mathbf{y}_{{\mathbb{S}},m}})).
	\end{equation}\
	
	2) Select measurements of the longest uniform part $\mathbb{U}$ from $\overline{\mathbf{x}}_{\mathbb{D},m}$, which is expressed as $\overline{\mathbf{x}}_{\mathbb{U},m}$. Then, construct the coarray covariance matrix as
	\begin{equation}
	\label{coarraycc}
	\langle\widetilde{\mathbf{R}}_{m}\rangle_{n_1,n_2} = \langle\overline{\mathbf{x}}_{\mathbb{U},m}\rangle_{n_1-n_2},
	\end{equation}
	where $n_1,n_2\in\mathbb{U}^+$ and $\mathbb{U}^+\subset\mathbb{U}$ is construct by all non-negative elements in $\mathbb{U}$.
	
	3) Perform eigen-decomposition of $\widetilde{\mathbf{R}}_{m}$, select $\textrm{size}(\mathbb{U}^+) - K$ smallest eigenvalue corresponding eigenvectors as column vectors to construct the noise subspace $\mathbf{U}_n$, where $\textrm {size}(\cdot)$ is the operator to get the size of sets.
	
	4) Compute the $K$ values of $\overline{\theta}$ which maximize the MUSIC spectrum proposed as follow
	\begin{equation}
	\label{musicsp11}
	P_m(\overline{\theta}) = \frac{1}{\mathbf{v}^H_{\mathbb{U}^+}(\overline{\theta})\mathbf{U}_n\mathbf{U}^H_n\mathbf{v}_{\mathbb{U}^+}(\overline{\theta})},
	\end{equation}
	then, the $K$ values of $\overline{\theta}_k$ are the estimated DOA.
\end{algorithm}

The flowchart of OB-MUSIC1 is shown in Algorithm \ref{algorithm1}.
Frankly speaking, we give OB-MUSIC1, which is similar to SS-MUSIC used in one-bit scalar sensor arrays \cite{liu2017one}, to make our motivation of OB-MUSIC2 more clear.
OB-MUSIC1 verifies that the DOAs of EM sources can be estimated by treating dipoles on any axis as a scalar array.
However, one-bit scalar arrays is not powerful enough to estimate DOAs of EM sources since the signal power is partially dropped, and the final DOA estimate also depends on the selection of axis.
To reveal this shortcoming, we rewrite (\ref{CPdecomposition}) as
$$\mathbf{R}_{\mathbf{s}_k} = \frac{p^2_k(1-\eta_k)}{2}\mathbf{I}_2+p^2_k\eta_k{\bm {\mathcal{J}}_{k,c}}{\bm {\mathcal{J}}^H_{k,c}},$$
where ${\bm {\mathcal{J}}_{k,c}}$ is defined as in (\ref{JonesVector}) with subscript $c$ denotes the CP part of the $k$-th signal. To this end, we have
$$p^2_{k,1} = p^2_k[(1-\eta_k)/2 + \eta_k\cos^2\varphi_k],$$
$$p^2_{k,2} = p^2_k[(1-\eta_k)/2 + \eta_k\sin^2\varphi_k],$$
with $p^2_{k,1}+p^2_{k,2} = p^2_k$.
It is clear that $p^2_k$ is divided into $p^2_{k,1}$ and $p^2_{k,2}$.
The scalar array on $x$-axis and $y$-axis can only use $p^2_{k,1}$ and $p^2_{k,2}$, respectively.
Especially on $y$-axis, the received signal power is multiplied by a real number $b^2_{2}(\overline{\theta}_k)$ less than $1$.
If $\overline{\theta}_k$ is close to $\pm\frac{1}{2}$, the performance on $y$-axis will be attenuated severely.

Specifically, the division of the power is determined by the unknown $\eta_k$ and $\varphi_k$.
For instance, we consider the scalar array made by dipoles on $x$-axis ($m=1$).
Then we define the power losses of the $k$-th signal received by the scalar array as $l_{k,1} = 10\log_{10}\frac{p^2_{k}}{p^2_{k,1}} =  -10\log_{10}[(1-\eta_k)/2 + \eta_k\cos^2\varphi_k]$.
We find that $l_{k,1}$ is a monotonically increasing function of both $\eta_k$ and $\varphi_k$ when $\varphi_k \in [\pi/4,\pi/2]$.
If $\eta_k > 0.9$ and  $\varphi_k > 4\pi/9$, we will have $l_{k,1}>11.1 \textrm{dB}$.
We also have the similar results for $l_{k,2} = 10\log_{10}\frac{p^2_{k}}{p^2_{k,2}}$ ($m=2$) in a given range of $\eta_k$ and $\varphi_k$.
This property of $l_{k,m}$ indicates that the performance of the scalar array is strongly related to $\eta_k$ and $\varphi_k$.
Specifically, for a CP signal with $\eta_k = 1$, if $\varphi_k > 2\pi/5$, we will have $l_{k,1}>10.2 \textrm{dB}$.

The power losses can be improved by constructing a summation matrix
\begin{equation}
\label{summation}
	\widetilde{\mathbf{R}} = \sum_{m=1}^{2}\widetilde{\mathbf{R}}_{m}.
\end{equation}
With the summation, the signal power received is $p^2_{k,sum} = p^2_{k,1}+b^2_{2}(\overline{\theta}_k)p^2_{k,2}$.
In comparison with $p^2_{k,1}$ and $p^2_{k,2}$, $p^2_{k,sum}$ is less sensitive to $\eta_k$ and $\varphi_k$, even though it partially relies on the DOA.
A qualitative explanation is that when the DOA is identifiable (i.e. $\overline{\theta}_k$ is not nearly to $\pm\frac{1}{2}$ and $p^2_{k}$ is large enough), the power losses of one axis can be countervailed by the power on another axis.
For instance, $\overline{\theta}_k$ is set to be $0.4$, if $[\eta_k,\varphi_k]$ are set as $[0.9,4\pi/9]$, $[1,4\pi/9]$ and $[1,17\pi/36]$, we will have $[l_{k,sum},l_{k,x}]$ approximately equal to $[3.8 \textrm{dB},11.1 \textrm{dB}]$, $[4.2 \textrm{dB},15.2 \textrm{dB}]$ and $[4.3 \textrm{dB},21.1 \textrm{dB}]$, respectively, where $l_{k,sum}$ is the power losses after summation defined as $l_{k,sum} = 10\log_{10}\frac{p^2_{k}}{p^2_{k,sum}}$.

After the summation, the coarray measurement $\overline{\mathbf{x}}_{\mathbb{U}}$ is constructed by selecting the longest uniform part $\mathbb{U}$ from $\overline{\mathbf{x}}_{\mathbb{D}}$ which is expressed as
\begin{equation}
\label{sumcomeas}
	\overline{\mathbf{x}}_{\mathbb{D}} = \mathbf{W}^\dag \textrm{vec}(\sum_{k=1}^Kp^2_{k,sum}\mathbf{v}_{{\mathbb{S}}}(\overline{\theta}_k)\mathbf{v}_{{\mathbb{S}}}^H(\overline{\theta}_k)+\sigma^2_{sum}\mathbf{I}_L),
\end{equation}
where $p^2_{k,sum} = \sum_{m=1}^{2}\frac{p^2_{k,m}b^2_m(\overline{\theta}_k)}{N_m}$ and $\sigma^2_{sum}=\sum_{m=1}^{2}\frac{\sigma^2}{N_m}$.
The steering vectors in (\ref{sumcomeas}) are the same as in $\overline{\mathbf{x}}_{\mathbb{D},m}$. The DOAs can still be estimated by SS-MUSIC.
To this end, we propose OB-MUSIC2 to estimate DOAs in Algorithm \ref{algorithm2}.

\begin{algorithm}
	\label{algorithm2}
	\caption{OB-MUSIC2}
	\KwIn{$\mathbf{R}_{\mathbf{y}_{{\mathbb{S}},m}}$, number of signals $K$, array configration $\mathbb{D}$}
	\KwOut{DOA of signals $ \overline{\theta}_k $}
	
	1) Construct the coarray covariance matrices $\widetilde{\mathbf{R}}_{m}$ as in Algorithm \ref{algorithm1}, then make a summation as
	\begin{equation}
		\label{sumr}
		\widetilde{\mathbf{R}}=\widetilde{\mathbf{R}}_{1}+\widetilde{\mathbf{R}}_{2}.
	\end{equation}
	
	2) Employ MUSIC on $\widetilde{\mathbf{R}}$, with spectrum been proposed as
	\begin{equation}
	\label{musicsp12}
	P(\overline{\theta}) = \frac{1}{\mathbf{v}^H_{\mathbb{U}^+}(\overline{\theta})\mathbf{U}_n\mathbf{U}^H_n\mathbf{v}_{\mathbb{U}^+}(\overline{\theta})},
	\end{equation}
	thus, all of $K$ DOAs are estimated.
\end{algorithm}

The computational complexity of OB-MUSIC2 is comparable as the SS-MUSIC for unquantized measurements in \cite{liu2015remarks}.
The additional part of OB-MUSIC2 is the covariance matrix reconstruction showed in (\ref{obc2}), taking $O(ZN)$ operations, where $Z$ is the number of snapshots and $N$ is defined as $N=(\left| \mathbb{D} \right|+1 )/2$.
As the complexity of SS-MUSIC is $O(ZN+N^3)$ \cite{liu2015remarks}, we find that the total complexity of OB-MUSIC2 is $O(ZN+ZN+N^3)$.
The complexity of proposed method is still dominated by the eigen-decomposition, which requires $O(N^3)$ computations.

\section{Cramer-Rao bound}
\label{CRB}
This section derives the CRB for the proposed DOA estimation under the assumptions made in section \ref{SM}.
The expression of CRB for unquantized cross-dipole array is also shown as a comparison.
Furthermore, we discuss the effect of one-bit measurements on DOA estimation based on the CRBs.

Let us rewrite the received signals in (\ref{dipolesmodel}) as
\begin{equation}
	\label{receiveSignalCRB}
	\mathbf{x}_{\mathbb{T}}(t)=\mathbf{A}_{\mathbb{T}}\mathbf{s}_{\mathbb{T}}(t)+\mathbf{n}_{\mathbb{T}}(t),
\end{equation}
where $$
\begin{aligned}
	\mathbf{x}_{\mathbb{T}}(t) &= [\mathbf{x}^T_{{\mathbb{S}},1}(t),\mathbf{x}^T_{{\mathbb{S}},2}(t)]^T,\\
	\mathbf{A}_{\mathbb{T}} &=
	\begin{bmatrix}
	\mathbf{A}_{{\mathbb{S}}}&0\\
	0&\mathbf{A}_{{\mathbb{S}}}
	\end{bmatrix},\\
	\mathbf{s}_{\mathbb{T}}(t)&=
	\begin{bmatrix}
	\overline{\mathbf{B}}_1\mathbf{s}_{\mathbb{K},1}(t)\\
	\overline{\mathbf{B}}_2\mathbf{s}_{\mathbb{K},2}(t)
	\end{bmatrix},\\
	\mathbf{n}_{\mathbb{T}}(t)&=[\mathbf{n}^T_{{\mathbb{S}},1}(t),\mathbf{n}^T_{{\mathbb{S}},2}(t)]^T.
\end{aligned}
$$
The one-bit measurements can be described as
 \begin{equation}
 \label{OBreceiveSignalCRB}
 \mathbf{y}_{\mathbb{T}}(t)=\textrm {signe}(\mathbf{x}_{\mathbb{T}}(t)).
 \end{equation}
Then the set of deterministic but unknown parameters to be estimated by one snapshot $z$ is $\mathbf{\Theta} = \left\lbrace \overline{\theta}_k, A_{k,1},A_{k,2},\kappa_{k,1},\kappa_{k,2} \right\rbrace_{k=1}^K $, where $A_{k,m}$ is the magnitude of the $k$th element in $\overline{\mathbf{B}}_m\mathbf{s}_{\mathbb{K},m}(z)$ and $\kappa_{k,m}$ is the phase.

The probability mass function (PMF) of $\mathbf{y}_{\mathbb{T}}(z)$ measured by one snapshot $z$, denoted as $p(\mathbf{y}_z\arrowvert\mathbf{\Theta})$, is expressed as
 \begin{equation}
 	\label{PMF}
 	p(\mathbf{y}_z\arrowvert\mathbf{\Theta})=\prod_{l=1}^{2L}p(\mathcal{R}\left\lbrace [\mathbf{y}_z]_l\right\rbrace \arrowvert\mathbf{\Theta})p(\mathcal{I}\left\lbrace [\mathbf{y}_z]_l\right\rbrace \arrowvert\mathbf{\Theta})
 \end{equation}
 where
 \begin{equation}
 	\begin{aligned}
 	p(\mathcal{R}\left\lbrace [\mathbf{y}_z]_l\right\rbrace \arrowvert\mathbf{\Theta}) &= \mathbb{P}(\mathcal{R}\left\lbrace [\mathbf{y}_z]_l\right\rbrace =1\arrowvert\mathbf{\Theta})^{\frac{1+\mathcal{R}\left\lbrace [\mathbf{y}_z]_l\right\rbrace}{2}}\\
 	\times &\mathbb{P}(\mathcal{R}\left\lbrace [\mathbf{y}_z]_l\right\rbrace =-1\arrowvert\mathbf{\Theta})^{\frac{1-\mathcal{R}\left\lbrace [\mathbf{y}_z]_l\right\rbrace}{2}}
 	\end{aligned}.
 \end{equation}

Let $\mathbf{r}_z$ and $\mathbf{i}_z$ denote the real and imaginary parts of $\mathbf{A}_{\mathbb{T}}\mathbf{s}_{\mathbb{T}}(z)$.
Since $\mathcal{R}\left\lbrace [\mathbf{x}_z]_l\right\rbrace \sim \mathcal{N}(\mathbf{r}_z,\frac{1}{2}\sigma^2)$ then
\begin{equation}
	\begin{aligned}
	\mathbb{P}(\mathcal{R}\left\lbrace [\mathbf{y}_z]_l\right\rbrace =1\arrowvert\mathbf{\Theta})&=\mathbb{P}(\mathcal{R}\left\lbrace [\mathbf{x}_z]_l\right\rbrace \geq 0\arrowvert\mathbf{\Theta})\\
	&=\Phi(\frac{[\mathbf{r}_z]_l}{\sigma})\\
	\mathbb{P}(\mathcal{R}\left\lbrace [\mathbf{y}_z]_l\right\rbrace =-1\arrowvert\mathbf{\Theta})&=\mathbb{P}(\mathcal{R}\left\lbrace [\mathbf{x}_z]_l\right\rbrace < 0\arrowvert\mathbf{\Theta})\\
	&=1-\Phi(\frac{[\mathbf{r}_z]_l}{\sigma})
	\end{aligned}
\end{equation}
where $\Phi(x) = \frac{1}{\sqrt{\pi}}\int_{-\infty}^{x}e^{-t^2}dt$. $p(\mathcal{I}\left\lbrace [\mathbf{y}_z]_l\right\rbrace \arrowvert\mathbf{\Theta})$ can be derived similar with $\mathbf{i}_z$ as above.

Using the results in \cite{fu2018quantized} and \cite{xi2020gridless}, the Fisher Information Matrix (FIM) for one-bit data measured by snapshot $z$ is
\begin{equation}
	\mathbf{I}_z(\mathbf{\Theta}) = \sum_{l=1}^L(\mathbf{I}_{z,l}^R(\mathbf{\Theta})+\mathbf{I}_{z,l}^I(\mathbf{\Theta}))	
\end{equation}
where
\begin{equation}
	\begin{aligned}
	\label{I}
	\mathbf{I}_{z,l}^R(\mathbf{\Theta})&=\frac{2}{\sigma^2}\omega\left( \frac{[\mathbf{r}_z]_l}{\sigma}\right) \left( \frac{\partial[\mathbf{r}_z]_l}{\partial\mathbf{\Theta}}\right) \left( \frac{\partial[\mathbf{r}_z]_l}{\partial\mathbf{\Theta}}\right) ^T\\
	\mathbf{I}_{z,l}^I(\mathbf{\Theta})&=\frac{2}{\sigma^2}\omega\left( \frac{[\mathbf{i}_z]_l}{\sigma}\right) \left( \frac{\partial[\mathbf{i}_z]_l}{\partial\mathbf{\Theta}}\right) \left( \frac{\partial[\mathbf{i}_z]_l}{\partial\mathbf{\Theta}}\right) ^T.
	\end{aligned}
\end{equation}
with $\omega(x)=\frac{\textmd{exp}(-2x^2)}{2\pi\Phi(x)[1-\Phi(x)]}$.

For all $Z$ snapshots, the FIM is given as
\begin{equation}
	\mathbf{I}(\mathbf{\Theta}) = \sum_{z=1}^Z\mathbf{I}_z(\mathbf{\Theta}).
\end{equation}
Then, the CRB of DOAs using one-bit measurements is obtained by taking the first $k$ diagonal elements from  $\mathbf{I}^{-1}(\mathbf{\Theta})$.

Additionally, the FIM without quantization is given as
\begin{equation}
\mathbf{\overline{I}}(\mathbf{\Theta}) = \sum_{z=1}^Z\mathbf{\overline{I}}_z(\mathbf{\Theta}).
\end{equation}
where
$$
\mathbf{\overline{I}}_z(\mathbf{\Theta}) = \sum_{l=1}^L(\mathbf{\overline{I}}_{z,l}^R(\mathbf{\Theta})+\mathbf{\overline{I}}_{z,l}^I(\mathbf{\Theta}))
$$
and
\begin{equation}
\begin{aligned}
\label{Ibar}
\mathbf{\overline{I}}_{z,l}^R(\mathbf{\Theta})&=\frac{2}{\sigma^2}\left( \frac{\partial[\mathbf{r}_z]_l}{\partial\mathbf{\Theta}}\right) \left( \frac{\partial[\mathbf{r}_z]_l}{\partial\mathbf{\Theta}}\right) ^T\\
\mathbf{\overline{I}}_{z,l}^I(\mathbf{\Theta})&=\frac{2}{\sigma^2} \left( \frac{\partial[\mathbf{i}_z]_l}{\partial\mathbf{\Theta}}\right) \left( \frac{\partial[\mathbf{i}_z]_l}{\partial\mathbf{\Theta}}\right) ^T.
\end{aligned}
\end{equation}
Then, the CRB of DOAs without quantization is obtained by taking the first $k$ diagonal elements from  $\mathbf{\overline{I}}^{-1}(\mathbf{\Theta})$.

Compared with (\ref{Ibar}), the FIM for one-bit measurements in (\ref{I}) is a weighted version of that for unquantized measurements.
The weight function, playing an important role on the FIM, is expressed as $\omega\left( \frac{[\mathbf{r}_z]_l}{\sigma}\right)$ for the real part and $\omega\left( \frac{[\mathbf{i}_z]_l}{\sigma}\right)$ for the imaginary part.
We make the following comments based on the property of $\omega(x)$ shown in Fig. \ref{omegaFigure}:

 \begin{figure}[h]
	\centering
	\includegraphics[scale=0.5]{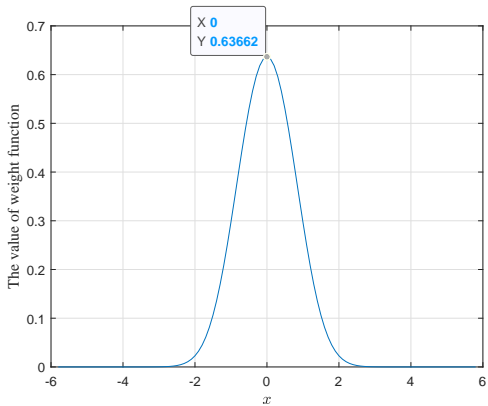}
	\caption{The weight function $\omega(x)$}
	\label{omegaFigure}
\end{figure}
\begin{description}
	\label{CRBanalysis}
	\item[1)] The upper bound of $\omega(x)$ is $\frac{2}{\pi}\approx0.6366$, which means $\mathbf{I}(\Theta)\preceq\frac{2}{\pi}\mathbf{\overline{I}}(\mathbf{\Theta})$.
	$\omega(x)$ achieves the upper bound if and only if $[\mathbf{r}_z]_l=[\mathbf{i}_z]_l=0$ for all $z$ and $l$, which means the signals are zero in finite SNR scenarios.
	As this, the information loss caused by one-bit measurements is more than $2\text{dB}$.
	\item[2)] When the signals are fixed, $w(x)$ is a decrease function as the noise power $\sigma^2$ increased, which means the information loss in one-bit measurements is larger in low SNR regime than in the high SNR regime.
\end{description}

\section{Independence Analysis of EM Signals}
\label{OREMS}
The critical assumption used in the proposed method is that all $K$ signals are independent. This assumption has been widely acknowledged in scalar array processing, but in vector-sensor array processing, the assumption about independence of signals have three aspects:
1) rewrite (\ref{sigvector}) as
\begin{equation}
	\label{sigvector3}
	\begin{aligned}
	\widehat{\mathbf{s}}_\mathbb{K}(t) =& [s_1(t),...,s_{K_{cp}}(t),s_{K_{cp}+1,1}(t),\\
							  &s_{K_{cp}+1,2}(t),...,s_{K,1}(t),s_{K,2}(t)]
	\end{aligned}
\end{equation}
where $0\leq K_{cp}\leq K$ is the number of CP signals, and then assume all the $2K-K_{cp}$ elements in (\ref{sigvector3}) are independent as in \cite{li1994efficient}; 2) the covariance matrix of (\ref{sigvector3}) is assumed to be of full rank as in \cite{ho1999estimating}; 3) the same as our assumption as in \cite{he2017direction}.
However, the first one can not fully represent the EM signals especially the PP signals, for example, if the $i$-th EM wave is modulated by a summation of two circular polarization signals in the DST mode as in \cite{nehorai1994vector}, $s_{i,1}$ and $s_{i,2}$ in (\ref{sigvector3}) will not be independent.
On the other hand, the second one can not work with sparse arrays and the arcsin law.
In fact, our assumption is weaker than the second one but stronger than the first one.
It is suitable for sparse arrays and the arcsin law. But its ability on representing EM signals has not been theoretically analyzed to our best knowledge.
To this end, we validate that our assumption is suitable in communication.

For the DST method in communication, the envelope of an EM wave can be regarded as the summation of two spatially orthogonal signals as in \cite{nehorai1994vector} and in \cite{schwartz1995communication} $$\mathbf{s}(t) = \mathbf{s}_A(t)+\mathbf{s}_B(t)=\mathbf{w}_As_A(t)+\mathbf{w}_Bs_B(t),$$
where $s_A(t)$ and $s_B(t)$ are transmitted signals which can be assumed to be zero-mean independent Gaussian processes as in \cite{brillinger1985maximum}. As this, we have
\begin{equation}\nonumber
\begin{aligned}
\mathbf{w}_A,\mathbf{w}_B\in&\mathbb{C}^{2\times1},\;\mathbf{w}_A^H\mathbf{w}_B=0,\; |\mathbf{w}_A| = |\mathbf{w}_B	| =1,\\
&\textrm {and}\;\mathbb{E}\left\lbrace s_A(t)s^*_B(t)  \right\rbrace =0.
\end{aligned}
\end{equation}
Then we can rewrite the EM signal vector as
\begin{equation}
	\label{safterop1}
	\begin{aligned}
	\mathbf{s}(t) & = [\mathbf{w}_A,\mathbf{w}_B][s_{A}(t),s_{B}(t)]^T\\
	&= \mathbf{\Gamma}\mathbf{\Lambda}[\overline{s}_{A}(t),\overline{s}_{B}(t)]^T,\\
	\mathbf{\Gamma}&=
	\begin{bmatrix}
	\cos\varphi&\sin\varphi\\
	-\sin\varphi e^{j\psi}&\cos\varphi e^{j\psi}
	\end{bmatrix},\\
	\mathbf{\Lambda} &=
	\begin{bmatrix}
	p_{A}e^{j\phi_A}&0\\
	0&p_{B}e^{j\phi_B}
	\end{bmatrix},
	\end{aligned}
\end{equation}
where $\mathbf{\Gamma}$ is the orthogonal basis normalized by $[\mathbf{w}_A,\mathbf{w}_B]$, $\overline{s}_{A}(t)$ and $\overline{s}_{B}(t)$ are normalized signals, $p^2_{A}$ and $p^2_{B}$ are signal power of $s_{A}(t)$ and $s_{B}(t)$, respectively, $\phi_A$ and $\phi_B$ denote the phase changed by normalizing.
In other hands, for the SST method, the signals is give as in (\ref{JonesVector}), which is a special case of (\ref{safterop1}) with $p^2_{A}=0$ or $p^2_{B}=0$.
Without loss of generality, we assume $p^2_{B}=0$. As this, $\overline{s}_{B}(t)$ does not contribute to the signal vector.

The discussions above are about one EM wave. In the following, we consider $K$ EM waves impinging on a sensor array.
The number of signals transmitted by $K$ EM waves is denoted as $q$, where $K\leq q \leq 2K$.
It means that $2(q-K)$ signals are transmitted with the DST model and the remaining $2K-q$ ones are using SST model.
If $q= 2K$, all EM waves are PP, and all the $q$ transmitted signals can be expressed as
\begin{equation}
\label{sqt}
\begin{aligned}
\mathbf{s}_{Q}(t) =& [p_{1,A}\overline{s}_{1,A}(t),p_{1,B}\overline{s}_{1,B}(t),...,\\
					&p_{K,A}\overline{s}_{K,A}(t),p_{K,B}\overline{s}_{K,B}(t)]^T
\end{aligned}
\end{equation}
where all $\overline{s}_{k,A}(t)$ and $\overline{s}_{k,B}(t)$ obey zero-mean i.i.d. Gaussian distribution.
If $q<2K$, there will be $2K-q$ CP signals.
We have $p_{i,A}>0$ and $p_{i,B}=0$ for any CP signal $\mathbf{s}_i(t)$.
Without loss of generality, we can assume $\overline{s}_{i,B}(t)$ is also a zero-mean normalized Gaussian process which is independent to all the other signals.

With the expression of $\mathbf{s}_{Q}(t)$, we find that all $2K$ signals $\overline{s}_{1,A}(t),\overline{s}_{1,B}(t),...,\overline{s}_{K,A}(t),\overline{s}_{K,B}(t)$ obey zero-mean i.i.d. complex Gaussian distribution.
Then we obtain the second order statistical property of $\mathbf{s}_{\mathbb{K}}(t)$ in (\ref{sigvector}) by the following theorem which declares that all $K$ EM signals are independent in communication.
\begin{theorem}
\label{theorem1}
Assume that $\overline{s}_{k,A}(t)$ and $\overline{s}_{k,B}(t)$ for all $k\in\mathbb{K}$ obey zero-mean i.i.d. complex Gussian distribution. $\mathbf{s}_{\mathbb{K}}(t)$ follows complex Gaussian distribution $\mathcal{CN}(0,\textrm {diag}(\mathbf{R}_{s_1},..,\mathbf{R}_{s_K}))$.
\end{theorem}
\begin{proof}
According to (\ref{safterop1}), we can express $\mathbf{s}_k(t)$ in communication as
\begin{equation}
	\mathbf{s}_k(t) = \mathbf{\Gamma}_k\mathbf{\Lambda}_k[\overline{s}_{k,A}(t),\overline{s}_{k,B}(t)]^T,
\end{equation}
where $\mathbf{\Gamma}_k$ and $\mathbf{\Lambda}_k$ are defined similar as in (\ref{safterop1}), $\overline{s}_{k,A}(t)$ and $\overline{s}_{k,B}(t)$ are transmitted signals by the $k$th EM wave.

As $\mathbf{s}_k(t)$ is a summation of two independent signals following zero-mean i.i.d. complex Gussian distribution, it follows multi-variate complex Gussian distribution, and we $\mathbb{E}\lbrace \mathbf{s}_k(t) \rbrace = [0,0]^T$. We also have $\mathbb{E}\lbrace \mathbf{s}_k(t)\mathbf{s}^H_k(t) \rbrace = \mathbf{R}_{s_k}$ by definition.

Consider $p,q \in \mathbb{K}$ and $p \neq q$, we have
$$
\mathbb{E}\lbrace \mathbf{s}_{p}(t)\mathbf{s}^H_{q}(t)\rbrace =\mathbf{\Gamma}_p \mathbf{\Lambda}_p\mathbb{E}\lbrace \mathbf{\Upsilon} \rbrace \mathbf{\Lambda}^H_q \mathbf{\Gamma}_q^H,
$$
where
$$
\mathbf{\Upsilon} =
\begin{bmatrix}
\overline{s}_{p,A}(t)\overline{s}^H_{q,A}(t)&\overline{s}_{p,A}(t)\overline{s}^H_{q,B}(t)\\
\overline{s}_{p,B}(t)\overline{s}^H_{q,A}(t)&\overline{s}_{p,B}(t)\overline{s}^H_{q,B}(t)
\end{bmatrix}.
$$
As $\overline{s}_{k,A}(t)$ and $\overline{s}_{k,B}(t)$ for all $k\in\mathbb{K}$ obey zero-mean i.i.d. complex Gussian distribution, we have $\mathbb{E}\lbrace \overline{s}_{p,A}(t)\overline{s}^H_{q,A}(t) \rbrace = \mathbb{E}\lbrace \overline{s}_{p,A}(t)\overline{s}^H_{q,B}(t) \rbrace = \mathbb{E}\lbrace \overline{s}_{p,B}(t)\overline{s}^H_{q,A}(t) \rbrace = \mathbb{E}\lbrace \overline{s}_{p,B}(t)\overline{s}^H_{q,B}(t) \rbrace = 0$.
Then we have
$\mathbb{E}\lbrace \mathbf{\Upsilon} \rbrace = \mathbf{0}_2$ and $\mathbb{E}\lbrace \mathbf{s}_{p}(t)\mathbf{s}^H_{q}(t)\rbrace = \mathbf{0}_2$.

Applying the statistical property of all $\mathbf{s}_k(t)$ into $\mathbf{s}_{\mathbb{K}}(t)$ in (\ref{sigvector}), we have $\mathbb{E}\lbrace \mathbf{s}_{\mathbb{K}}(t) \rbrace = 0$ and $\mathbb{E}\lbrace \mathbf{s}_{\mathbb{K}}(t)\mathbf{s}^H_{\mathbb{K}}(t) \rbrace = \textrm {diag}(\mathbf{R}_{s_1},...,\mathbf{R}_{s_K})$. Proof is complete.
\end{proof}

\section{Numerical Results}
\label{NR}
In this section, we provide numerical results for the DOA estimation performance of the proposed methods.
All the results are obtained from 5000 independent Monte-Carlo experiments.
Unless otherwise specified, the arrays configuration in the experiments are shown as below with 10 cross-dipoles.
\begin{equation}
\label{arrayconfiginnr}
\begin{aligned}
\mathbb{S}_u &= \lbrace 0,1,2,3,4,5,6,7,8,9\rbrace,\\
\mathbb{S}_n &= \lbrace 1,2,3,4,5,6,12,18,24,30\rbrace,\\
\mathbb{S}_c &= \lbrace 0,3,5,6,9,10,12,15,20,25\rbrace,
\end{aligned}
\end{equation}
where $\mathbb{S}_u$, $\mathbb{S}_n$ and $\mathbb{S}_c$ denote the configuration of ULA, nested array and coprime array, respectively. These mean that $L_1 = L_2 = 5$ in nested array and $M = 3$ and $N=5$ in coprime array. In the experiments, the EM sources have unit power with $p^2_k=1$ and known number, $\varphi$ and $\psi$ are random variables following the uniform distribution in their domain of definition. The noises are equal power for all dipoles with variance $\sigma^2$ and the signal-to-noise ratio (SNR) is defined as $\textrm {SNR}=10\log_{10}\frac{\sum_{k=1}^Kp^2_k}{2k\sigma^2}=10\log_{10}\frac{1}{2\sigma^2}$.

The received signals are sampled by finite snapshots, and the covariance matrix is estimated by the sampled covariance matrix as
\begin{equation}
\label{infinitecm}
\mathbf{R}_{\mathbf{y}_{\mathbb{S},m}} \approx \widetilde{\overline{\mathbf{R}}}_{\mathbf{y}_{\mathbb{S},m}} = \frac{1}{Z}\sum_{z=1}^{Z}\mathbf{y}_{\mathbb{S},m}(z)\mathbf{y}^H_{\mathbb{S},m}(z),
\end{equation}
where $Z$ is the number of snapshots.

\subsection{Performance of Proposed Method}

First, we illustrate the MUSIC spectrum of two proposed methods with the nested array and the coprime array. We consider $15$ EM signals impinging on the two sparse one-bit cross-dipoles arrays. The signals are assumed to be sent from locations of normalized angles uniform distributed in $[-0.4,0.4]$. The DOPs of all signals are $0.5$. The SNR is set to be $10\textrm {dB}$. We use $200$ snapshots in this experiment.
Fig. \ref{spectrum} shows the MUSIC spectrum. Pictures in different rows correspond to different methods. Results from nested and coprime array are on the left and right column, respectively.
As we all know, nested array constructed in (\ref{arrayconfiginnr}) can identify 29 sources, while the coprime array can resolve 17 sources.
We use $15$ sources in order to be comparable, $15$ is bigger than the number of sensors which is $10$.
It is seen that all two scalar arrays using OB-MUSIC1 and the cross-dipole array using OB-MUSIC2 can identify sources more than sensors.
In this experiment, the DOP equals 0.5, which is not close to 1.
It has weak influence on the performance of OB-MUSIC1, as is revealed in Section \ref{OCSADOAE}.
Therefore, both OB-MUSIC1 and OB-MUSIC2 have the similar MUSIC spectrum, even though there exists some difference among peak heights of all methods.


 \begin{figure}[h]
	\centering
	\includegraphics[scale=0.5]{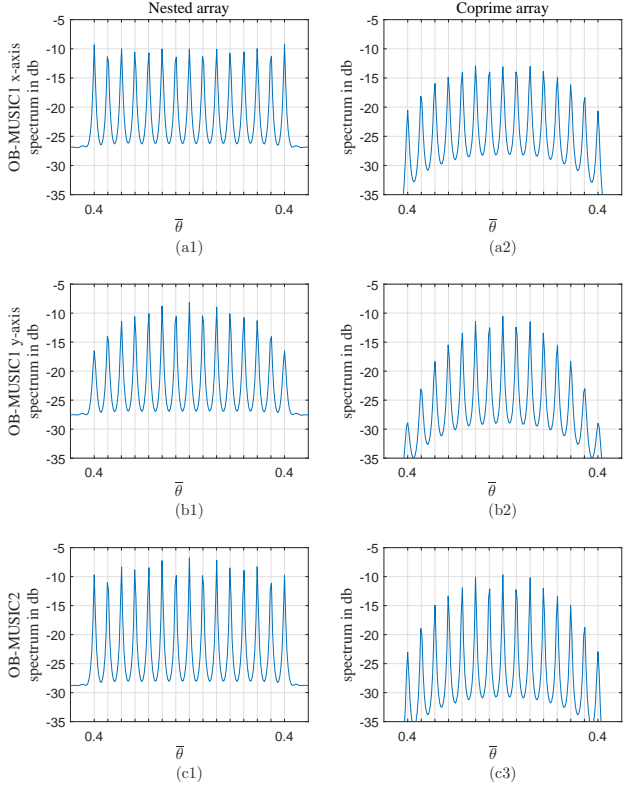}
	\caption{MUSIC spectrum of proposed methods with nested and coprime one-bit cross-dipoles array. (a) OB-MUSIC1 on $x$-axis, (b) OB-MUSIC1 on $y$-axis and (c) OB-MUSIC2. The (*1) are of nested array and (*2) are of coprime. Number of sources $K=15$, $\textrm {SNR} = 10\textrm {dB}$, number of snapshots $Z = 200$, 5000 Monte-Carlo runs}
	\label{spectrum}
\end{figure}

To explore the effect of DOP on OB-MUSIC1, we show the second experiment with the DOPs varying.
In this experiment, the number of sources is $K=5$, the DOA is $-0.4+0.2(k-1)$ for the $k$-th source.
The DOPs of all sources are the same, and they vary from $0$ to $1$ with step $0.1$.
The auxiliary polarization angle $\varphi_k$ is uniformly distributed in $[0,\pi/2]$.
SNR is set as $10\textrm {dB}$, and the number of snapshots is $200$.
The estimated DOAs are obtained by root-MUSIC. The performance is quantized by mean squared error (MSE) defined as $\textrm {MSE} = \sum_{k=1}^{K}(\widehat{\overline{\theta}}_k-\overline{\theta}_k)^2$.
For convenience, we abbreviate OB-MUSIC1 and OB-MUSIC2 as OB1 and OB2 in the following figures, respectively.
As shown in Fig. \ref{rou}, the performance of OB-MUSIC1 degenerates severely when DOPs increase. In contrary, OB-MUSIC2 is robust against DOPs.
When $\eta = 0$, the power is evenly distributed on the $x$ and $y$ axises.
But as $\eta$ increases, energy received by dipoles on each axis is more and more random because of the random $\varphi_k$, leading to the SNR of dipoles on each axis unstable.
Especially, the performance of OB-MUSIC1 on $y$ axis degrades faster than that on $x$ axis, because the energy impinging on $y$ axis is multiplied by an additional positive real number $b^2_{2}(\overline{\theta}_k)$ which is less than $1$.
This result confirms our statement in Section \ref{OCSADOAE}.

\begin{figure}[h]
	\centering
	\includegraphics[scale=0.5]{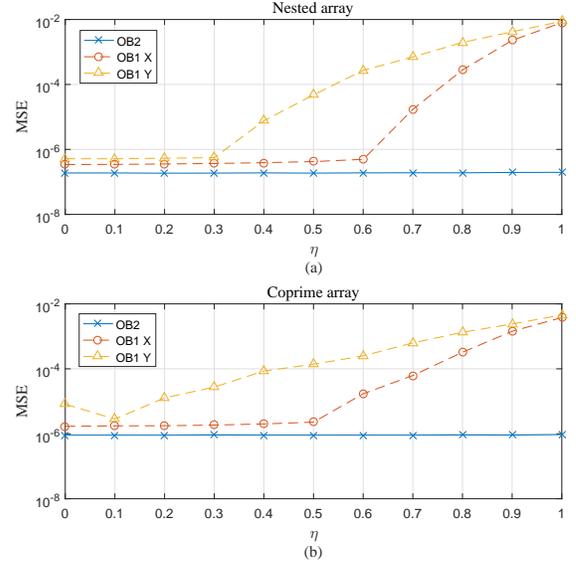}
	\caption{MSE versus DOP $\eta$, (a) on nested array, (b) on coprime array. Number of sources $K=5$, $\textrm {SNR} = 10\textrm {dB}$, number of snapshots $Z = 200$, 5000 Monte-Carlo runs.}
	\label{rou}
\end{figure}

Next, we verify the robustness of OB-MUSIC2 to PP and CP signals with varying $\text{SNR}$.
The DOPs of PP sources are random variables following uniform distribution $\mathcal{U}(0,0.99)$.
The DOPs of CP sources are $1$.
The $\text{SNR}$ vary from $-10\text{dB}$ to $20\text{dB}$ with step $5$.
The number of snapshots is set to $200$.
Fig. \ref{Rou_SNR} shows that the performance of OB-MUSIC2 is comparable for CP and PP signals on nested and coprime array when $\text{SNR}>-5\text{dB}$ and $\text{SNR}>0\text{dB}$ respectively.
However, in the low $\text{SNR}$ regime, the performance for PP signals is better than that for CP signals.
As we all know, the performance of SS-MUSIC will decrease obviously with the decreasing of $\text{SNR}$, if $\text{SNR}$ is lower than a threshold.
In Section \ref{OCSADOAE}, we have revealed that CP signals have bigger power loses than PP signals, leading that CP signals have lower total $\text{SNR}$s than PP signals.
As this, the performance of PP signals is better in the low $\text{SNR}$ regime, but in the high $\text{SNR}$ regime, the performance of the two kinds of signals is similar.

\begin{figure}[h]
	\centering
	\includegraphics[scale=0.5]{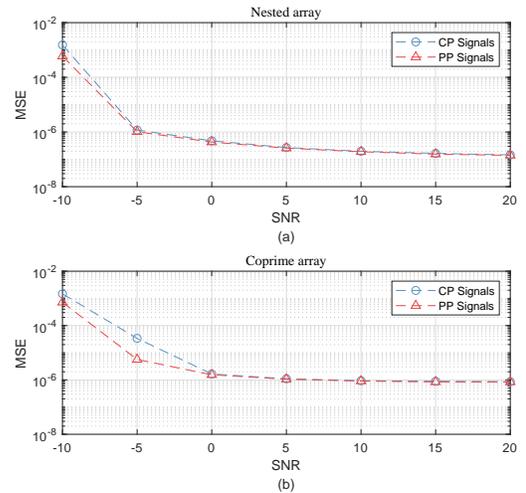}
	\caption{Performance of OB-MUSIC2 for PP and CP signals versus SNR, (a) on nested array, (b) on coprime array. Number of sources $K=5$, number of snapshots $Z = 200$, 5000 Monte-Carlo runs.}
	\label{Rou_SNR}
\end{figure}

Finally, we examine the increase of DOF on sparse arrays.
In this experiment, the array configurations are set to $\mathbb{S}_n = \lbrace 1,2,3,4,8,12\rbrace,$ and $\mathbb{S}_c = \lbrace 0,2,3,4,6,9\rbrace,$ which can detected up to 11 and 7 sources respectively \cite{pal2010nested, pal2011coprime}.
Here, we use fewer dipoles to make the results in Fig. \ref{DOFp} more intuitive.
The $\text{SNR}$ is set to be $10\text{dB}$.
The number of snapshots is $1000$.
The DOPs are set to be random variables following uniform distribution.
Fig. \ref{DOFp} shows that the OB-MUSIC2 successfully resolve 11 sources and 7 sources on the nested array and the coprime array respectively.

\begin{figure}[h]
	\centering
	\includegraphics[scale=0.5]{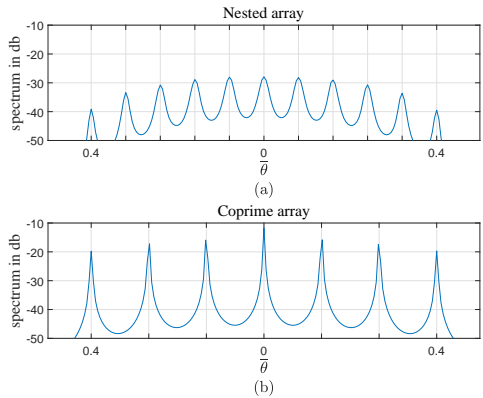}
	\caption{MUSIC spectrum of OB-MUSIC2 with nested and coprime one-bit cross-dipoles array, (a) 11 sources on nested array, (b) 7 sources on coprime array.  $\textrm {SNR} = 10\textrm {dB}$, number of snapshots $Z = 1000$, 5000 Monte-Carlo runs.}
	\label{DOFp}
\end{figure}

\subsection{One-bit measurements vs Unquantized measurements}

We now compare the performance between one-bit measurements and the unquantized one.
In the experiments, the sources are the same as that in Fig. \ref{rou} except DOPs.
The DOPs are set to be random variables following uniform distribution, since DOPs are usually unknown in applications.
OB-MUSIC1 will not be shown because its performance is much sensitive to DOPs.
The performances of these methods on ULA are also demonstrated as a comparison.
As a benchmark, the CRB of both one-bit and unquantized measurements are also provided.
The method in \cite{he2017direction} is used for unquantized measurements.
Although the method in \cite{he2017direction} was only developed for nested array, it can be easily extended to be suitable for coprime array by dropping the data which are out of the longest uniform part on the difference coarray.
Furthermore, we use SS-MUSIC developed in \cite{liu2015remarks} taking the place of SS-MUSIC step in \cite{he2017direction} to reduce the computational complexity.


\begin{figure*}[h]
	\centering
	\includegraphics[scale=0.5]{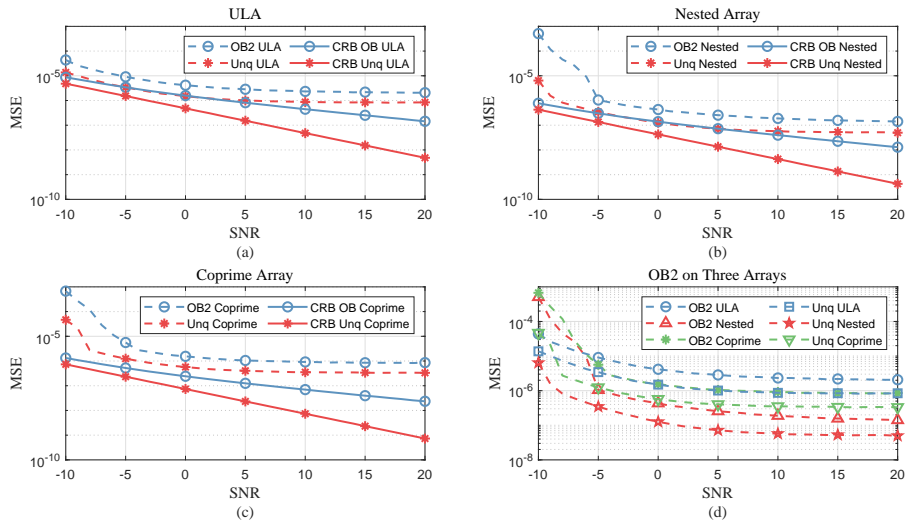}
	\caption{MSE of OB-MUSIC2 and CRB versus the SNR. Number of sources $K=5$, snapshots $Z = 200$, 5000 Monte-Carlo runs.}
	\label{SNR}
\end{figure*}

In Fig. \ref{SNR}, we compare the DOA estimation performance on the three arrays with varying SNR.
The number of snapshots is $200$. 
The quantization loss is defined as $10\log_{10}(\textrm{MSE}_{one-bit}/\textrm{MSE}_{unquantized})$ for performance metric.
In Fig. \ref{SNR} (a), (b) and (c), when the SNR is higher, the gap between the CRB of one-bit and unquantized measurments increases on all the three arrays, which validate the comments in Section \ref{CRB}.
In Fig. \ref{SNR} (a) and (b), when $-5 \text{dB}<\text{SNR}<10 \text{dB}$ the performance of SS-MUSIC on ULA and nested array using unquantized measurments is nearly to the CRB of one-bit measurments.
However, in Fig. \ref{SNR} (c), the performance of unquantized measurments on coprime can not reach the CRB of one-bit measurments within the same $\text{SNR}$ range.
The reason is that we have dropped some data by selecting measurements of the longest uniform part on the coprime array, but none data has been dropped on ULA or nested arrays.
In Fig. \ref{SNR} (d), when $\textrm{SNR} >-5 \text{dB}$, the one-bit nested array has better performance than the unquantized coprime array, and one-bit coprime array has comparable performance to the unquantized ULA.
Based on this observation, we find that one-bit sparse cross-dipole arrays provide a compromise between the DOA estimation performance and the system complexity.
When $\textrm{SNR}=0\textrm{dB}$, the quantization losses are $5.3\textrm{dB}$, $4.3\textrm{dB}$ and $4.4\textrm{dB}$ for the nested array, the coprime array and ULA, respectively.
But as SNR becomes smaller than $-5\textrm{dB}$, the performance of one-bit measurements deteriorates faster than that of unquantized ones, especially on the nested array, for example, when $\textrm{SNR}=-10 \textrm{dB}$, the quantization losses are $19.0\textrm{dB}$ and $11.6\textrm{dB}$ for nested and coprime array, respectively.
Interestingly, the quantization loss is robust to SNR on ULA.
When $\textrm{SNR} <-6 \textrm{dB}$, the performance of ULA is better than that of sparse arrays.
These results indicate that one-bit sparse cross-dipole arrays increase DOFs at the expense of the reduced anti-noise performance.

\begin{figure*}[h]
	\centering
	\includegraphics[scale=0.5]{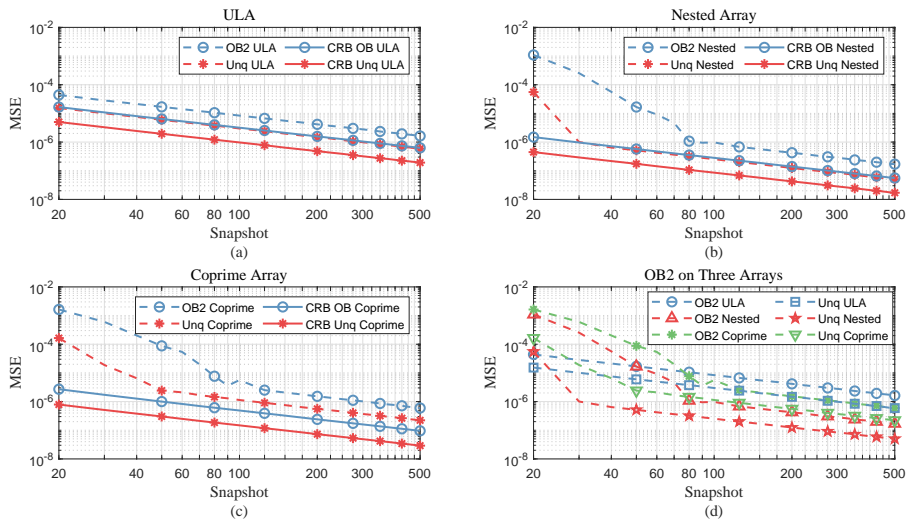}
	\caption{MSE of OB-MUSIC2 and CRB versus the snapshots. Number of sources $K=5$, $\textrm{SNR} = 0\textrm{dB}$, 5000 Monte-Carlo runs.}
	\label{snapshots}
\end{figure*}
Fig. \ref{snapshots} shows the MSE of the proposed method and the CRB on three arrays with varying snapshots.
SNR is set as $0\textrm{dB}$.
The MSE of three arrays and corresponding CRBs are shown in Fig. \ref{snapshots} (a), (b) and (c).
The MSE of unquantized measurements by SS-MUSIC on ULA and nested array is nearly the same as the CRBs of one-bit measurements on these arrays when the snapshots $Z>30$.
However, due to discarding the non-uniform part of coarray, there is a gap between MSE of SS-MUSIC and the one-bit measurement CRB on the coprime array when $Z>30$.
When the number of snapshots $Z$ is larger than $80$, the one-bit nested array has better performance than the unquantized coprime array.
When $Z\geq125$, the one-bit coprime array has comparable performance with the unquantized ULA.
The quantization losses are almost stable when $Z\geq100$.
For instance, when $Z=[125,500]$, the quantization losses are $[5.3\textrm{dB},5.3\textrm{dB}]$, $[4.4\textrm{dB},4.3\textrm{dB}]$ and $[4.4\textrm{dB},4.4\textrm{dB}]$ for the nested array, the coprime array and ULA, respectively.
When $Z$ reduces from $80$ to $20$, the quantization losses on sparse arrays increase fast, but it is robust on ULA. For instance, when $Z=40$, the quantization losses are $19.4\textrm{dB}$, $14.9\textrm{dB}$ and $4.5\textrm{dB}$ for the nested array, the coprime array and ULA, respectively.
This phenomenon may stem from the statistical efficiency of recovered covariance. 
As stated in \cite{wang2016coarrays}, the SS-MUSIC for sparse array via difference coarray decreases the statistical efficiency, while the covariance reconstruction from one-bit samples also depends on the statistical property of original unquantized covariance.
With the reduced snapshots, statistical efficiency of recovered covariance is not enough to cover the cost of coarray SS-MUSIC.
To reduce this gap, we will seek to develop more powerful methods to estimate the parameters without the difference coarray or the covariance reconstruction in the future works.

\begin{figure}[h]
	\centering
	\includegraphics[scale=0.5]{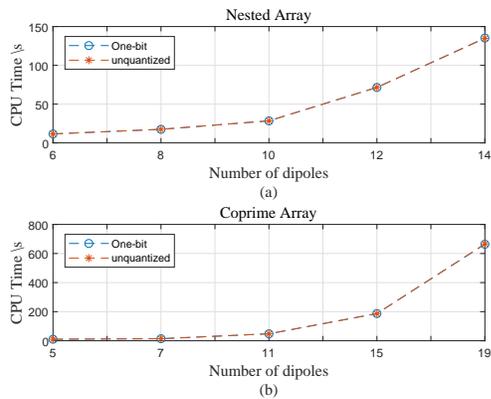}
	\caption{CPU time versus number of dipoles for two arrays, (a) nested array, (b) coprime array.}
	\label{computationalComplexity}
\end{figure}

In the last experiment, we compare the computational complexity of OB-MUSIC2 with that of SS-MUSIC in \cite{he2017direction}.
The number of snapshots is set to $200$.
SNR is set to $0\textrm{dB}$.
We use the nested arrays with $\left[ L_1,L_2\right] \in\left( \left[3,3\right],\left[4,4\right],\left[5,5\right],\left[6,6\right],\left[7,7\right] \right) $ and the coprime array with $\left[ M,N\right] \in\left( \left[2,3\right],\left[3,4\right],\left[5,6\right],\left[7,8\right],\left[9,10\right] \right) $ where $L_1,L_2,M,N$ are defined in (\ref{nestedarrays}) and (\ref{coprimearrays}) respectively.
On a Windows 10 workstation with two Intel Xeon E5-2660v2 cores and 64 GB RAM, 5000 Monte-Carlo runs has been taken without parallel computing.
As shown in Fig. \ref{computationalComplexity}, the CPU running time of the proposed method is almost the same as the SS-MUSIC in \cite{he2017direction}, which means the covariance matrix reconstruction has no significant effect on the computational complexity.
Hence the complexity is dominated by the eigen-decomposition in both methods.

\section{Conclusion}
\label{Conclusion}
This paper proposed a one-bit measurement scheme for cross-dipoles sparse array used to estimating DOAs of EM signals.
We presented the DOA estimation method based on SS-MUSIC with robust performance on solving both PP and CP signals.
We also derive the CRB of DOA estimation.
For the critical assumption that signals are independent, we theoretically validated its reasonability in communication.
The assumption is the prerequisite of the proposed and other typical DOA estimation methods for EM signals.
Numerical results revealed that the quantization loss stemming from one-bit measurement is stable when SNR and the number of snapshots are larger than a threshold (e.g. the threshold is $\textrm{SNR} = -5 \textrm{dB}$ and $Z = 100$ in our experiments).
More importantly, one-bit sparse cross-dipoles arrays have comparable performance to the unquantized ULA with same sensors, and therefore provide a compromise between the DOA estimation performance and the system complexity.

\bibliographystyle{IEEEtran}
\bibliography{doa}
\end{document}